\pdfoutput=1
\RequirePackage{ifpdf}
\ifpdf 
\documentclass[pdftex]{sigma}
\else
\documentclass{sigma}
\fi

\newcommand{\mbZ}{\mathbb Z}
\newcommand{\oh}{\overline h}
\newcommand{\hLambda}{\widehat\Lambda}
\def\mbQ{{\mathbb Q}}
\def\d{{\partial}}
\newcommand{\eps}{\varepsilon}
\newcommand{\hcA}{\widehat{\mathcal A}}
\DeclareMathOperator{\res}{res}
\newcommand{\mcL}{\mathcal{L}}
\newcommand{\hL}{\widehat{L}}
\newcommand{\ILW}{\mathrm{ILW}}
\newcommand{\hA}{\widehat{A}}
\newcommand{\hp}{\widehat{p}}
\newcommand{\Lax}{\mathrm{Lax}}

\numberwithin{equation}{section}

\begin{document}

\allowdisplaybreaks

\newcommand{\arXivNumber}{1809.00271}

\renewcommand{\PaperNumber}{120}

\FirstPageHeading

\ShortArticleName{Simple Lax Description of the ILW Hierarchy}

\ArticleName{Simple Lax Description of the ILW Hierarchy}

\Author{Alexandr BURYAK~$^{\dag\ddag}$ and Paolo ROSSI~$^\S$}

\AuthorNameForHeading{A.~Buryak and P.~Rossi}

\Address{$^\dag$~School of Mathematics, University of Leeds, Leeds, LS2 9JT, UK}
\EmailD{\href{mailto:a.buryak@leeds.ac.uk}{a.buryak@leeds.ac.uk}}
\URLaddressD{\url{http://sites.google.com/site/alexandrburyakhomepage/home/}}

\Address{$^\ddag$~Faculty of Mechanics and Mathematics, Lomonosov Moscow State University,\\
\hphantom{$^\ddag$}~Moscow, GSP-1, 119991, Russia}

\Address{$^\S$~Dipartimento di Matematica ``Tullio Levi-Civita'', Universit\`a degli Studi di Padova,\\
\hphantom{$^\S$}~Via Trieste 63, 35121 Padova, Italy}
\EmailD{\href{mailto:paolo.rossi@math.unipd.it}{paolo.rossi@math.unipd.it}}
\URLaddressD{\url{http://www.math.unipd.it/~rossip/}}

\ArticleDates{Received September 07, 2018, in final form November 06, 2018; Published online November 10, 2018}

\Abstract{In this note we present a simple Lax description of the hierarchy of the intermediate long wave equation (ILW hierarchy). Although the linear inverse scattering problem for the ILW equation itself was well known, here we give an explicit expression for all higher flows and their Hamiltonian structure in terms of a single Lax difference-differential operator.}


\Classification{37K10}

\section{Introduction}

The intermediate long wave (ILW) equation was introduced in \cite{Jos77} to describe the propagation of waves in a two-layer fluid of finite depth. The model represents the natural interpolation between the Benjamin--Ono deep water and Kortweg--de Vries shallow water theories. The ILW equation for the time $s$ evolution of a function $w=w(x)$ defined on the real line has the following form:
\begin{gather*}
w_s +2 w w_x + T(w_{xx})= 0, \qquad T(f) := \mathrm{p.v.}\!\int_{-\infty}^{+\infty}\! \frac{1}{2\delta}\left( \operatorname{sgn}(x-\xi)-\coth\frac{\pi (x-\xi)}{2\delta} \right) f(\xi) {\rm d}\xi,
\end{gather*}
where $\delta$ is a parameter and $\mathrm{p.v.}\int\cdot \, {\rm d}\xi$ represents the principal value integral. This equation can be rewritten in the formal loop space formalism (see for instance \cite{BDGR16b}, from which we borrow notations for the rest of the paper) as
\begin{gather}\label{eq:ILW equation}
u_t = u u_x + \sum_{g\geq 1} \mu^{g-1} \eps^{2g} \frac{|B_{2g}|}{(2g)!}u_{2g+1},
\end{gather}
where $B_{2g}$ are the Bernoulli numbers. Indeed the action of the operator $T$ on a function $f$ can be written in terms of the derivatives of $f$ as $T(f) = \sum\limits_{n \geq 1} \delta^{2n-1} 2^{2n} \frac{|B_{2n}|}{(2n)!} \d_x^{2n-1}f$, and, by setting $w = \frac{\sqrt{\mu}}{\eps} u$, $s = - \frac{\eps}{2 \sqrt{\mu}} t$ and $\delta = \frac{\eps\sqrt{\mu} }{2}$, we see how the above equations are equivalent.

In \cite{SAK79} the integrability of the ILW equation was established, finding an infinite number of conserved densities and giving a corresponding inverse scattering problem. From this inverse scattering problem a Lax representation can of course be deduced. Also, through a simple Hamiltonian structure, the conserved densities generate an infinite number of commuting flows. However the explicit relation between the Lax representation of the ILW equation, its higher flows and the Hamiltonian structure was never, to the best of our knowledge, clarified in the literature.

In this paper we give an explicit and remarkably simple Lax description of the full hierarchy of commuting flows of the ILW equation (the ILW hierarchy) and their Hamiltonian structure in terms of a single Lax mixed difference-differential operator.

As it turns out, this Lax description corresponds to the equivariant bi-graded Toda hierarchy of \cite[Section~B.2]{MT11} in the somewhat degenerate case of bi-degree $(1,0)$, hence a reduction of the 2D-Toda hierarchy. In fact, in \cite{MT11}, because of the geometric origin of their problem, the authors always work with strictly positive bi-degree, but this is an unnecessary restriction. What we do here is prove that the bi-degree $(1,0)$ case corresponds to the ILW hierarchy in its Hamiltonian formulation. A similar question is raised in \cite[Remark~31]{Car03}, where the author notices that the definition of all but the extended flows of the extended bi-graded Toda hierarchy survives when one of the bi-degrees vanishes. Since in the equivariant case there is no need of extended flows, this problem does not arise in the Lax representation of the equivariant bi-graded Toda hierarchy. Of course this means that equivariant bi-graded Toda hierarchies of any bidegree~$(N,0)$, with $N\geq 1$, are well defined (see also \cite{BCRR14} for their relation with the rational reductions of the 2D-Toda hierarchy). We will study the relation of such hierarchies with the geometry of equivariant orbifold Gromov--Witten theory in an upcoming publication.

\section{ILW hierarchy}

Here and in what follows we will use the formal loop space formalism in the notations of \cite{BDGR16b}. The Hamiltonian structure of the ILW equation~\eqref{eq:ILW equation} is given by the Hamiltonian
\begin{gather*}
\oh_1^\ILW=\int\left(\frac{u^3}{6}+\sum_{g\ge 1}\mu^{g-1}\eps^{2g}\frac{|B_{2g}|}{2(2g)!}u u_{2g}\right){\rm d}x
\end{gather*}
and the Poisson bracket $\{\cdot,\cdot\}_{\d_x}$, associated to the operator $\d_x$. The Hamiltonians $\oh_d^{\ILW}\in\hLambda_u^{[0]}$, $d\ge 2$, of the higher flows of the ILW hierarchy are uniquely determined by the properties
\begin{gather*}
\left.\oh_d^\ILW\right|_{\eps=0}=\int\frac{u^{d+2}}{(d+2)!}{\rm d}x,\qquad\left\{\oh_d^\ILW,\oh_1^\ILW\right\}_{\d_x}=0.
\end{gather*}
For example,
\begin{gather*}
\oh^{\ILW}_2=\int\left(\frac{u^4}{4!}+\frac{\eps^2}{48}u^2 u_{xx}+\sum_{g\ge 2}\frac{|B_{2g}|}{(2g)!}\eps^{2g}\left(\mu^{g-2}\frac{g+1}{2}u u_{2g}+\mu^{g-1}\frac{1}{4}u^2u_{2g}\right)\right){\rm d}x.
\end{gather*}
We refer the reader to the paper~\cite[Section 8]{Bur15}, which explains a relation between the Hamilto\-nians~$\oh^\ILW_d$ and the conserved quantities of the ILW equation, constructed in~\cite{SAK79}. It is convenient to introduce an additional Hamiltonian $\oh_0^\ILW:=\int\frac{u^2}{2}{\rm d}x$, which generates spatial translations. So the flows of the ILW hierarchy are given by
\begin{gather*}
\frac{\d u}{\d t_d}=\d_x\frac{\delta\oh_d^{\ILW}}{\delta u},\qquad d\ge 0,
\end{gather*}
where we identify the times $t_1$ and $t$.

\section{Lax description of the ILW hierarchy}

Our Lax description of the ILW hierarchy is presented in Section~\ref{subsection:Lax description}, see Theorem~\ref{main theorem}. Before that, in Section~\ref{subsection:shift operators}, we recall necessary definitions from the theory of shift operators.

\subsection{Shift operators}\label{subsection:shift operators}

Let $\Lambda:={\rm e}^{{\rm i}\eps\d_x}$. We will consider formal series of the form
\begin{gather*}
A=\sum_{n\le m}a_n\Lambda^n,\qquad a_n\in\hcA_u,\qquad m\in\mbZ.
\end{gather*}
Via the operation of composition $\circ$, the vector space of such formal operators is endowed with the structure of a non-commutative associative algebra. The positive part~$A_+$, the negative part~$A_-$ and the residue~$\res A$ of the operator $A$ are defined by
\begin{gather*}
A_+:=\sum_{n=0}^m a_n\Lambda^n,\qquad A_-:=A-A_+,\qquad \res A:=a_0.
\end{gather*}
Let $z$ be a formal variable. The symbol $\hA$ of the operator $A$ is defined by
\begin{gather*}
\hA:=\sum_{n\le m}a_n {\rm e}^{nz}.
\end{gather*}

For an operator $L$ of the form
\begin{gather*}
L=\Lambda+\sum_{n\ge 0}a_n\Lambda^{-n},\qquad a_n\in\hcA_u,
\end{gather*}
one can define the dressing operator $P$,
\begin{gather*}
P=1+\sum_{n\ge 1}p_n\Lambda^{-n},
\end{gather*}
by the identity
\begin{gather*}
L=P\circ\Lambda\circ P^{-1}.
\end{gather*}
Note that the coefficients $p_n$ of the dressing operator do not belong to the ring $\hcA_u$, but to a~certain extension of it~(see, e.g.,~\cite[Section~2]{CDZ04}). The dressing operator $P$ is defined up to the multiplication from the right by an operator of the form $1+\sum\limits_{n\ge 1}\hp_n\Lambda^{-n}$, where $\hp_n$ are some constants.

The logarithm $\log L$ is defined by
\begin{gather*}
\log L:=P\circ {\rm i}\eps\d_x\circ P^{-1}={\rm i}\eps\d_x-{\rm i}\eps P_x\circ P^{-1},
\end{gather*}
where $P_x=\sum\limits_{n\ge 1}(p_n)_x\Lambda^{-n}$. The ambiguity in the choice of dressing operator is cancelled in the definition of $\log L$ and, moreover, the coefficients of $\log L$ do belong to $\hcA_u$ (see the proof of Theorem~2.1 in~\cite{CDZ04}). To be more precise, one has the commutation relations
\begin{gather}\label{eq:computation of log L}
\big[\log L,L^m\big]=0,\qquad m\ge 1,
\end{gather}
which imply that
\begin{gather*}
\res \big[{\rm i}\eps P_x\circ P^{-1},L^m\big]={\rm i}\eps\d_x\res L^m,\qquad m\ge 1.
\end{gather*}
These relations allow to compute recursively all the coefficients of the operator ${\rm i}\eps P_x\circ P^{-1}$. As a result, if we write
\begin{gather*}
\log L={\rm i}\eps\d_x+\sum_{n\ge 1}f_n\Lambda^{-n},
\end{gather*}
then the coefficient $f_n$ can be expressed as a differential polynomial in the coefficients $a_0,a_1,\ldots$, $a_{n-1}$ of the operator~$L$, $f_n=f_n(a_0,\ldots,a_{n-1})\in\hcA^{[0]}_{a_0,\ldots,a_{n-1}}$. For example,
\begin{gather}\label{eq:formula for f_1}
f_1=\frac{{\rm i}\eps\d_x}{\Lambda-1}a_0=\sum_{n\ge 0}\frac{B_n}{n!}({\rm i}\eps\d_x)^n a_0.
\end{gather}

\subsection{Lax description}\label{subsection:Lax description}

Let $\tau$ be a formal variable and
\begin{gather*}
\mcL:=\Lambda+u-\tau {\rm i}\eps\d_x.
\end{gather*}
From the discussion of the construction of the logarithm $\log L$ in the previous section it is easy to see that there exists a unique operator $L$ of the form
\begin{gather*}
L=\Lambda+\sum_{n\ge 0}a_n\Lambda^{-n},\qquad a_n\in\hcA^{[0]}_u[\tau],
\end{gather*}
satisfying
\begin{gather}\label{eq:definition of L}
L-\tau\log L=\mcL.
\end{gather}
Since $\big[L^{d+1},\mcL\big]=0$, $d\ge 0$, the commutator $\big[\big(L^{d+1}\big)_+,\mcL\big]$ doesn't contain terms with non-zero powers of $\Lambda$. Consider the following system of PDEs:
\begin{gather}\label{eq:Lax equations}
\frac{\d u}{\d T_d}=\frac{\d\mcL}{\d T_d}=\frac{1}{\tau {\rm i}\eps(d+1)!}\big[\big(L^{d+1}\big)_+,\mcL\big],\qquad d\ge 0.
\end{gather}
The following theorem is the main result of our paper.
\begin{theorem}\label{main theorem}\quad
\begin{enumerate}\itemsep=0pt
\item[$1.$] The flows $\frac{\d}{\d T_d}$, given by~\eqref{eq:Lax equations}, pairwise commute.
\item[$2.$] The system of Lax equations~\eqref{eq:Lax equations} possesses a Hamiltonian structure given by the Hamiltonians
\begin{gather*}
\oh^\Lax_d=\int\left(\frac{\res L^{d+2}}{(d+2)!}-\frac{\tau}{d+1}\frac{\res L^{d+1}}{(d+1)!}\right){\rm d}x,\qquad d\ge 0,
\end{gather*}
and the Poisson bracket associated to the operator $\d_x$.
\item[$3.$] Let $y$ be a formal variable and define polynomials $P_d(y)\in\mbQ[y,\tau]$, $d\ge 1$, by
\begin{gather*}
P_d(y):=y\prod_{i=1}^{d-1}\left(y+\frac{\tau}{{\rm i}}\right)=\sum_{j=1}^d P_{d,j}y^j\tau^{d-j},\qquad P_{d,j}\in\mbQ.
\end{gather*}
The ILW hierarchy is related to the hierarchy~\eqref{eq:Lax equations} by the following triangular transformation:
\begin{gather}\label{eq:relation between Hamiltonians}
\oh_d^\Lax=\sum_{j=0}^d P_{d+1,j+1}\tau^{d-j}\left.\oh_j^\ILW\right|_{\substack{\mu=-\tau^{-1}\\\eps\mapsto\eps\sqrt{-\tau}}},\qquad d\ge 0.
\end{gather}
\end{enumerate}
\end{theorem}
\begin{proof}1. Hereafter, for simplicity, we will use $L^m_+$ to denote $(L^m)_+$. Let
\begin{gather*}
H_d:=\frac{1}{\tau {\rm i}\eps(d+1)!}L^{d+1}_+,\qquad d\ge 0.
\end{gather*}
Let us first check that
\begin{gather}\label{eq:formulas for derivatives of L and log L}
\frac{\d L}{\d T_d}=[H_d,L],\qquad\frac{\d\log L}{\d T_d}=[H_d,\log L].
\end{gather}
Equations~\eqref{eq:computation of log L} and~\eqref{eq:definition of L} imply that
\begin{gather}
\frac{\d L}{\d T_d}-\tau\frac{\d\log L}{\d T_d}=[H_d,L]-\tau[H_d,\log L],\label{eq:equation1 for derivatives}\\
\res\left[\frac{\d\log L}{\d T_d},L^m\right]+\sum_{a+b=m-1}\res\left[\log L,L^a\circ\frac{\d L}{\d T_d}\circ L^b\right]=0,\qquad m\ge 1.\label{eq:equation2 for derivatives}
\end{gather}
We consider these equations as a system of equations for the pair of operators $\frac{\d L}{\d T_d},\frac{\d\log L}{\d T_d}$. Similarly to the discussion of the computation of the logarithm $\log L$ in the previous section, equations~\eqref{eq:equation1 for derivatives} and~\eqref{eq:equation2 for derivatives} allow to compute recursively all the coefficients of the operators $\frac{\d L}{\d T_d}$ and~$\frac{\d\log L}{\d T_d}$. Then it remains to note that the operators $\frac{\d L}{\d T_d}=[H_d,L]$ and $\frac{\d\log L}{\d T_d}=[H_d,\log L]$ satisfy system~\eqref{eq:equation1 for derivatives}--\eqref{eq:equation2 for derivatives}. This completes the proof of equations~\eqref{eq:formulas for derivatives of L and log L}.

When we know formulas~\eqref{eq:formulas for derivatives of L and log L}, the commutativity of the flows $\frac{\d}{\d T_d}$ is proved by a standard computation:
\begin{gather*}
 \tau {\rm i}\eps(d_1+1)!\tau {\rm i}\eps(d_2+1)!\left(\frac{\d}{\d T_1}\frac{\d u}{\d T_2}-\frac{\d}{\d T_2}\frac{\d u}{\d T_1}\right) \\
\qquad{} = \big[\big[L^{d_1+1}_+,L^{d_2+1}\big]_+,\mcL\big]+\big[L^{d_2+1}_+,\big[L^{d_1+1}_+,\mcL\big]\big]\\
\qquad\quad{} -\big[\big[L^{d_2+1}_+,L^{d_1+1}\big]_+,\mcL\big]-\big[L^{d_1+1}_+,\big[L^{d_2+1}_+,\mcL\big]\big]\\
\qquad {} \stackrel{\substack{\text{Jacobi}\\\text{identity}}}{=} \big[\big[L^{d_1+1}_+,L^{d_2+1}\big]_+,\mcL\big]-\big[\big[L^{d_2+1}_+,L^{d_1+1}\big]_+,\mcL\big]+\big[\big[L^{d_2+1}_+,L^{d_1+1}_+\big],\mcL\big]\\
\qquad{} = \underbrace{-\big[\big[L^{d_1+1}_-,L^{d_2+1}_+\big]_+,\mcL\big]-\big[\big[L^{d_2+1}_+,L^{d_1+1}\big]_+, \mcL\big]}_{=-\big[\big[L^{d_2+1}_+,L^{d_1+1}_+\big]_+,\mcL\big]}+\big[\big[L^{d_2+1}_+,L^{d_1+1}_+\big],\mcL\big]=0.
\end{gather*}
2. Note that the flows $\frac{\d}{\d T_d}$ can be written as
\begin{gather*}
\frac{\d u}{\d T_d}=\frac{1}{(d+1)!}\d_x\res L^{d+1}.
\end{gather*}
Let us compute the flow $\frac{\d}{\d T_1}$. For the coefficients of the operator $L$, one can immediately see that $a_0=u$ and then, using formula~\eqref{eq:formula for f_1}, we get
\begin{gather*}
a_1=\tau\frac{{\rm i}\eps\d_x}{\Lambda-1}u.
\end{gather*}
This allows to compute
\begin{gather*}
\frac{\d u}{\d T_1}= \frac{1}{2}\d_x\res L^2=\d_x\left(\frac{u^2}{2}+\frac{\tau {\rm i}\eps\d_x}{2}\frac{\Lambda+1}{\Lambda-1}u\right)=uu_x+\tau u_x-\tau\sum_{g\ge 1}\frac{|B_{2g}|}{(2g)!}\eps^{2g}u_{2g+1}\\
\hphantom{\frac{\d u}{\d T_1}}{} = \d_x\frac{\delta}{\delta u}\left(\frac{u^3}{6}+\tau\frac{u^2}{2}-\tau\sum_{g\ge 1}\frac{|B_{2g}|}{2(2g)!}\eps^{2g}u u_{2g}\right).
\end{gather*}

The local functionals $\oh_d^\Lax$ are conserved quantities for the flow $\frac{\d}{\d T_1}$. Indeed,
\begin{gather*}
\frac{\d}{\d T_1}\int\res L^d {\rm d}x=\frac{1}{2\tau {\rm i}\eps}\int\res\big[L^2_+,L^d\big]{\rm d}x,
\end{gather*}
which is zero because
\begin{gather*}
\int\res\big[f\Lambda^m,g\Lambda^n\big]{\rm d}x=\delta_{m+n,0}\int\big(f\cdot\Lambda^m g-g \cdot\Lambda^n f \big){\rm d}x=0,\qquad f,g\in\hcA_u,\qquad m,n\in\mbZ.
\end{gather*}
Therefore, the local functionals $\oh_d^\Lax$ together with the Poisson bracket $\{\cdot,\cdot\}_{\d_x}$ generate the flows which commute with the flow $\frac{\d}{\d T_1}$. Then these flows are uniquely determined by their dispersionless parts (see~\cite[Lemma 3.3]{LZ06} or~\cite[Lemma~4.14]{BR18}). Hence, it is sufficient to check the equation
\begin{gather*}
\d_x\frac{\delta\oh_d^\Lax}{\delta u}=\d_x\frac{\res L^{d+1}}{(d+1)!}
\end{gather*}
at the dispersionless level.

Denote $\hL_0:= \hL\big|_{\eps=0}$. We see that it is sufficient to check that
\begin{gather}\label{eq:Hamiltonians at the dispersionless level}
\d_x\frac{\d}{\d u}\left(\frac{\res\hL_0^{d+2}}{(d+2)!}-\frac{\tau}{d+1}\frac{\res\hL_0^{d+1}}{(d+1)!}\right)=\d_x\res\left(\frac{\hL_0^{d+1}}{(d+1)!}\right),\qquad d\ge 0.
\end{gather}
For this we compute
\begin{gather*}
\hL_0-\tau\log\hL_0={\rm e}^z+u-\tau z \ \Rightarrow \ \frac{\d\hL_0}{\d u}-\tau\frac{\frac{\d\hL_0}{\d u}}{\hL_0}=1 \ \Rightarrow \ \frac{\d\hL_0}{\d u}=\frac{1}{1-\tau\hL_0^{-1}}.
\end{gather*}
Therefore,
\begin{gather*}
\frac{\d}{\d u}\left(\frac{\res\hL_0^{d+1}}{(d+1)!}\right)=\frac{1}{d!}\res\left(\hL_0^d\frac{\d\hL_0}{\d u}\right)=\frac{1}{d!}\sum_{j=0}^d\tau^j\res\big(\hL_0^{d-j}\big),\qquad d\ge 0,
\end{gather*}
which gives
\begin{gather}\label{eq:recursion}
\frac{\d}{\d u}\left(\frac{\res\hL_0^{d+1}}{(d+1)!}\right)=\frac{\res\hL_0^d}{d!}+\frac{\tau}{d}\frac{\d}{\d u}\left(\frac{\res\hL_0^d}{d!}\right),\qquad d\ge 1.
\end{gather}
This implies equation~\eqref{eq:Hamiltonians at the dispersionless level}.

3. We see that
\begin{gather*}
\frac{\d u}{\d T_1}=\left.\frac{\d u}{\d t_1}\right|_{\substack{\mu=-\tau^{-1}\\\eps\mapsto\eps\sqrt{-\tau}}}+\tau u_x.
\end{gather*}
Using again the result of~\cite[Lemma~3.3]{LZ06} (see also~\cite[Lemma~4.14]{BR18}), we conclude that it is sufficient to prove equation~\eqref{eq:relation between Hamiltonians} at the dispersionless level, namely,
\begin{gather*}
\res\left(\frac{\hL_0^{d+2}}{(d+2)!}-\frac{\tau}{d+1}\frac{\hL_0^{d+1}}{(d+1)!}\right)=\sum_{j=0}^d P_{d+1,j+1}\tau^{d-j}\frac{u^{j+2}}{(j+2)!},\qquad d\ge 0.
\end{gather*}
Using formula~\eqref{eq:recursion} and the property $\hL_0\big|_{u=0}={\rm e}^z$, the last equation can be equivalently written as
\begin{gather}\label{eq:main elementary identity}
\frac{\res\hL_0^{d+1}}{(d+1)!}=\sum_{j=0}^d P_{d+1,j+1}\tau^{d-j}\frac{u^{j+1}}{(j+1)!}.
\end{gather}
Recursion~\eqref{eq:recursion} implies that
\begin{gather*}
\frac{\res\hL_0^{d+1}}{(d+1)!}=\left(\prod_{j=1}^d\left(\d_u^{-1}+\frac{\tau}{j}\right)\right)u,
\end{gather*}
where we define the action of the operator $\d_u^{-1}$ in the polynomial ring $\mbQ[u,\tau]$ by $\d_u^{-1}u^j:=\frac{u^{j+1}}{j+1}$, $j\ge 0$. Since we obviously have
\begin{gather*}
\sum_{j=0}^d P_{d+1,j+1}\tau^{d-j}\frac{u^{j+1}}{(j+1)!}=\left(\prod_{j=1}^d\left(\d_u^{-1}+\frac{\tau}{j}\right)\right)u,
\end{gather*}
identity~\eqref{eq:main elementary identity} becomes clear. This completes the proof of the theorem.
\end{proof}

\subsection*{Acknowledgements}

We would like to thank Andrea Brini, Guido Carlet, Oleg Chalykh, Allan Fordy, Alexander Mikhailov and Vladimir Novikov for useful discussions.
The work of the first author (Theorem~\ref{main theorem}, parts~1 and~3) was supported by the grant no.~16-11-10260 of the Russian Science Foundation.

\pdfbookmark[1]{References}{ref}
\LastPageEnding

\end{document}